\documentclass[11pt]{article}

\usepackage[T1]{fontenc}        
\usepackage{hyperref}
\usepackage{cite}
\usepackage{fullpage}
\usepackage{amsfonts}
\usepackage{amssymb}
\usepackage{amsmath}
\usepackage{amsthm}
\usepackage{xcolor}

\usepackage{booktabs}
\usepackage{pifont}
\newcommand{\cmark}{\ding{51}}%
\newcommand{\xmark}{\ding{55}}%

\newtheorem{theorem}{Theorem}
\newtheorem{claim}{Claim}
\newtheorem{lemma}{Lemma}
\newtheorem{definition}{Definition}
\newtheorem{conjecture}{Conjecture}

\newcommand{\omv}{\textnormal{\textsf{OMv}}}
\newcommand{\oumv}{\textnormal{\textsf{OuMv}}}
\usepackage{url}  

\DeclareMathOperator{\cost}{cost}

\begin{document}
\title{On the Complexity of
	Weight-Dynamic Network Algorithms%
		\thanks{Research supported by the Austrian Science Fund (FWF) 
		and netIDEE SCIENCE project P 33775-N and
		the Vienna Science and Technology Fund (WWTF) project
		ICT19-045 (WHATIF).}
}

\author{Monika Henzinger \quad Ami Paz \quad Stefan Schmid \\
{\small Faculty of Computer Science, University of Vienna}
}
\date{}

\maketitle

\begin{abstract}
While operating communication networks adaptively
may improve utilization and performance, frequent
adjustments also introduce an algorithmic challenge:
the re-optimization of traffic engineering solutions
is time-consuming and may limit the
granularity at which a network can be adjusted.
This paper is motivated by question whether the reactivity of a network
can be improved by re-optimizing solutions \emph{dynamically}
rather than \emph{from scratch}, especially if inputs such as link weights 
do not change significantly. 

This paper explores to what extent dynamic algorithms can be used
to speed up fundamental tasks in network operations. 
We specifically investigate optimizations
related to  traffic engineering (namely shortest paths and maximum flow
computations), but also consider spanning tree
and matching applications. 
While prior work on dynamic graph algorithms focuses on link insertions and deletions,
we are interested in the practical problem of \emph{link weight changes}.

We revisit existing upper bounds 
in the weight-dynamic model, and present several
novel lower bounds on the amortized runtime for
recomputing solutions.
In general, we find that the potential performance gains depend on the application, 
and there are also strict
limitations on what can be achieved, even if link weights change only slightly.
\end{abstract}

\section{Introduction}

As communication networks are often an expensive infrastructure,
making best use of the given resources is important.
For example, over the last decades, many efficient traffic engineering 
algorithms have been developed which allow ISPs to improve network utilization
and performance~\cite{vissicchio2014sweet,b4}.
These algorithms typically rely on the optimization of link weights
(in particular the IGP link weights),
which in turn determine the shortest paths computed by the 
ECMP protocol~\cite{ecmp}.

In principle, a more dynamic and adaptive operation of communication
networks has the potential to significantly improve
the network efficiency: traffic patterns often feature
a high degree of temporal structure~\cite{sigmetrics20complexity,benson2010understanding,pujol2019steering},
which could be exploited for optimizations over time.
Motivated by this potential, we have recently seen great efforts
to render networks more flexible, adaptive, or even ``self-driving''~\cite{pieee19,FeamsterR2017,reda2020path,fischer2006replex}.

This paper is motivated by the observation that a more dynamic network operation
also introduces an algorithmic challenge. Going back to our traffic engineering
example, the state-of-the-art approach of re-computing shortest paths or
optimal flow allocations 
 \emph{from scratch} can be time-consuming
\cite{pujol2019steering},
and may eventually become the bottleneck which limits network reactivity.
The high runtimes may further prevent operators to conduct fast what-if analyses,
e.g., experimenting with different link weights using tools such as FlowVisor~\cite{pujol2019steering}.
While such a high runtime may be unavoidable in some cases, 
in practice it is unlikely that many link weights need to be changed significantly in a short time~\cite{vanbever2011seamless}. 
This introduces an optimization opportunity, and we ask: 
\begin{itemize}
\item Given a small change of link weights, 
can we recompute a solution \emph{dynamically}, i.e.,
based on the current solution, significantly faster than recomputing it from scratch?
\end{itemize}
Link weight changes are not only an important operation in the context of 
traffic engineering and shortest path routing, but link weights also define
other fundamental network structures, such as spanning trees \cite{perlman1985algorithm}.

We are in the realm of dynamic graph algorithms, an active research area in
theoretical computer science~\cite{Henzinger18}. However, while existing literature in this
area primarily revolves around the question whether solutions to graph theoretical problems can be recomputed
efficiently upon \emph{additions and removals} of links, given the networking context, we are many times interested in quickly reacting to changes of the \emph{link weights} only.

\paragraph{Our contributions.}
Motivated by the desire to speed up traffic engineering decisions,
we present novel \emph{weight-dynamic graph algorithms},
network algorithms which dynamically recompute solutions upon link
weight changes, and
initiate the study of lower bounds on the time per operation for weight-dynamic graph algorithms. 
Our results are summarized in Table~\ref{table: results}.

Dynamic algorithms are usually composed of a database that is maintained during input changes. They present trade-offs between two parameters:
\emph{update time}, which is the time it takes to update the database upon a change (edge insertion or deletion, or edge-weight change);
and \emph{query time}, which is the time it take to answer a query using the database (e.g., compute the distance between a pair of nodes).
Our algorithms all have a constant query time, and different update times. Our lower bounds show that in some cases, either the update time or the query time must be high.

\begin{table*}
	\centering
	\small
	 \begin{tabular*}{\linewidth}{|l|lr|l|lr|c|}
		\toprule
		Problem 
		& \multicolumn{2}{l|}{Sta. upper bound} 
		& Dyn. upper bound~
		& \multicolumn{2}{l|}{Lower bound}
		&Cond.?\\ 
		\toprule
		Shortest $(s,t)$-path
		& $O(m+n\log n)$
		& \cite{Dijkstra59}
		&$O(\sqrt{m}\log n)$
		\hfill\cite{FrigioniMN00}	
		& 
		\begin{tabular}{@{}l@{}}
			$\Omega(n^{1-\epsilon})$ update \\
			$\Omega(n^{2-\epsilon})$ query
		\end{tabular}
		& Thm.~\ref{thm: shortest st path - lb} &\cmark
		\\
		\midrule
		Max $(s,t)$-flow
		&$O(m^{10/7}W^{1/7})$ 
		&\cite{Madry16}
		& $O(m)$
		\hfill Thm.~\ref{thm: mx st flow with only weight increases - alg}
		& 
		\begin{tabular}{@{}l@{}}
			$\Omega(n^{1-\epsilon})$ update \\
			$\Omega(n^{2-\epsilon})$ query
		\end{tabular}
		& Thm.~\ref{thm: mx st flow with only weight increases - lb} &\cmark
		\\
		\midrule
		Maximum matching
		& $O(m^{10/7})$ 
		& \cite{Madry16}
		& $O(m)$
		\hfill Thm.~\ref{thm: mwm - alg}
		& 
		\begin{tabular}{@{}l@{}}
			$\Omega(n^{1-\epsilon})$ update \\
			$\Omega(n^{2-\epsilon})$ query
		\end{tabular}
 		& Thm.~\ref{thm: mwm with only weight increases}
		&\cmark
		\\
				\midrule
		MST
		& $O(m)$
		& \cite{KargerPT95}
		& $O\left(\frac{\log^4 n}{\log \log n}\right)$ 
		\hfill \cite{HolmRW15}
		& 
		\begin{tabular}{@{}l@{}}
			$\Omega(\log n)$ update \\
			$\Omega(\log n)$ query
		\end{tabular}
		& Thm.~\ref{thm: mst - lb}
		&\xmark
		\\
		\bottomrule
	\end{tabular*}
	~
	\caption{Our results vs. some prior results on edge-dynamic algorithms. 
	\textnormal{\newline\small The dynamic upper bounds are for update times, and have constant query time. 
	The conditional lower bounds should be interpreted as follows: in a graph with $m=\Theta(n^2)$ edges there is no algorithm that has an update time of $O(n^{1-\epsilon})$ and a query of $O(n^{2-\epsilon})$ for any $\epsilon > 0$;
	the MST lower bound states there is no algorithm that has an update time of  $o(\log n)$ and a query of $o(\log n)$.
	The shortest paths lower bound holds even for $(5/3-\delta)$ approximations, for any $\delta>0$.
	$W$ denotes the maximum weight of an edge in the graph.}
}
	\label{table: results}
\end{table*}

We first consider fundamental problems related to routing,
and in particular, the computation of shortest paths and maximum
flows. 
Computing a shortest $(s,t)$-path for a new network can be done by a simple application of Dijkstra’s algorithm~\cite{Dijkstra59} in $O(m+n\log n)$ time.
This can be improved to $O(\sqrt{m}\log n)$ time per update in a dynamic network with only edge-weight changes~\cite{FrigioniMN00}.
	We prove that the above algorithm is almost optimal:
\begin{itemize}
	\item
	For any two constants $\epsilon,\delta>0$,
	the amortized update time for a $(5/3-\delta)$-approximation of 
	the shortest $(s,t)$-path problem is $\Omega(n^{1-\epsilon})$, even when edge weights can only change by a small additive constant 
	(this bound is consistent with the upper bound, as the lower bound is proved on a graph with $m=\Theta(n^2)$).
	This yields strong lower bounds on the ability to pre-compute the effect of edge-weight changes on the network flow, even when the changes are only by a constant factor
	and if a multiplicative error is allowed.
\end{itemize}

In terms of the maximum flow problem, we show: 
\begin{itemize}
	\item
	There is a simple algorithm that maintains the maximum $(s,t)$-flow in $O(m)$ time per change.
	\item 	
	The amortized update time for maximum $(s,t)$-flow is $\Omega(n^{1-\epsilon})$.
\end{itemize}

Motivated by these results, we 
then extend our study of weight-dynamic algorithms to
applications related to load balancing:
we consider different matching problems, including 
maximum weight matching, $b$-matching~\cite{AhnG14}, and 
semi-matching~\cite{HarveyLLT06}.
We extend the definition of semi-matchings to weighted graphs in a natural way, and show:
\begin{itemize}
	\item
	For any constant $\epsilon>0$,
	the amortized update time for a weighted maximum matching, semi-matching, or $b$-matching, 
	is $\Omega(n^{1-\epsilon})$.
\end{itemize}

The above lower bounds
hold under the popular assumption of the \omv{} conjecture~\cite{HenzingerKNS15} from complexity theory, discussed later.
Finally, we study another basic network problem, the computation of a minimum weight spanning tree (MST), for which we prove an unconditional lower bound, as follows.
\begin{itemize}
	\item
	There is an algorithm that maintains an MST in 
	$O(\log^4 n /\log \log n)$ amortized time~\cite{HolmRW15}.
	While this algorithm is for edge addition and deletion, it can be adapted to handle edge-weight changes in a trivial manner.
	\item
	The amortized update time of MST is $\Omega(\log n)$; this bound is unconditional.

\end{itemize}

All our lower-bound results hold for deterministic as well as for randomized algorithms with error probability at most~$1/3$, and also when amortized over a sequence of changes.

In general, we hope that our insights can inform the networking community what can and
cannot be achieved with such dynamic algorithms, and believe that the notion of 
``weight-dynamic'' network algorithm may also be of interest for future research in the theory community.

\paragraph{Organization.}
After defining the weight-dynamic model in Section~\ref{sec:prelim},
we consider its application for 
the shortest $(s,t)$-paths problem (Section~\ref{sec: sp}),
the maximum $(s,t)$-flow problem (Section~\ref{sec: max flow}),
different matching problems (Section~\ref{sec: matchings}),
and finally the spanning tree problem (Section~\ref{sec: mst}).
We review related work in Section~\ref{sec:relwork} and
conclude our contribution in Section~\ref{sec:conclusion}.

\section{Preliminaries}
\label{sec:prelim}

We consider dynamic graphs where only the \emph{edge weights} change, while the underlying structure of the graph remains intact.
We are particularly interested in the practically relevant scenario where
link weights do not change dramatically at once: an operator is unlikely
to adjust weights quickly and globally \cite{vanbever2011seamless}.
Moderate changes render the problem different: for example, in distance-related problems (shortest paths, MST), removing an edge in the ``standard'' model is equivalent to giving it a very high weight in our model.
In this paper, we consider a more restricted model, where between any two consecutive graphs, the edge weight differences are limited by an additive constant.

Interestingly, while dynamic graph algorithms are well explored in the more theoretical literature (see e.g.~\cite{Zwick_98,RodittyZ11,RodittyZ12,HenzingerKN16,HenzingerKN16a,HenzingerKN18,GutenbergW20,GutenbergW20a,GutenbergW20b} for work on exact and approximate shortest paths), 
most existing work considers scenarios where entire links change, and hence
these results do not directly apply to network optimization problems such as traffic engineering where 
only
the weights of links change. 
 
\subsection{Our model}
We consider a simple weighted graph $G=(V,E,w)$
with $n=|V|$ nodes, $m=|E|$ edges, and an edge-weight function $w:E\to\{1,\ldots,W\}$ giving each edge a positive integer weight bounded by some value $W$ (that may depend, e.g., on~$n$).

A fixed graph with dynamic edge weights is modeled as a sequences $G_1,G_2,...$ of graphs, finite or infinite, 
where each $G_i=(V,E,w_i)$ has the same topology but 
possibly a different edge-weight function.
All edge weight functions have the same domain, i.e., they are of the form $w_i:E\to\{1,\ldots,W\}$ for the same bound~$W$.

We focus on \emph{bounded edge-weight changes}:
for a given constant $c$, for any $i\geq 1$ and edge $e\in E$, we have $|w_{i+1}(e)-w_i(e)|\leq c$.
Our lower bounds hold also in a less strict regime, where the edge weight changes are not bounded by a constant, and can even be multiplicative.

\subsection{Problem definitions}
We consider several classic problems on weighted graphs, all having important applications in networking.

\paragraph{Shortest $(s,t)$-path}
An $(s,t)$-path in a graph $G=(V,E,w)$ connecting two nodes $s,t\in V$ is a sequence $P=(s=v_0,v_1,\ldots,v_k=t)$ of nodes, such that for each $0\leq i<k$, $(v_i,v_{i+1})\in E$. 
The weight (usually representing length) of such a path is $w(P)=\sum_{i=0}^{k-1}w(v_i,v_{i+1})$.
Given two nodes $s,t\in V$, a \emph{shortest $(s,t)$-path} is an $(s,t)$-path with minimal weight, i.e., such that no other $(s,t)$-path in $G$ has strictly lower weight.
The goal in the shortest $(s,t)$-path problem is to find such a path, or evaluate its length.

\paragraph{Maximum $(s,t)$-flow}
An $(s,t)$-flow in a graph $G=(V,E,w)$ is a function $f:V\times V\to \{0,\ldots,W\}$ assigning each directed edge an amount of flow. 
Formally, we require for each $u,v\in V$:
if $(u,v)\notin E$ then $f(u,v)=0$ (and specifically $f(v,v)=0$);
if $u\notin\{s,t\}$ then $\sum_{u'\in V}w(u,u') = \sum_{u'\in V}w(u',u)$;
and,
if $(u,v)\in E$ then $f(u,v)\leq w(u,v)$.
The \emph{value} (or amount) of such a flow is defined to be $\sum_{u'\in V}f(s,u')$.
We define the residual weight of each edge $e\in E$ as $w(e)-f(e)$, and the \emph{residual graph} as a graph with the same structure but with the residual weights.

Given two nodes $s,t\in V$, a \emph{maximum $(s,t)$-flow} is an $(s,t)$-flow of maximum value.
The goal in the maximum $(s,t)$-flow problem is to find such a flow its value.

\paragraph{Matchings}
A \emph{matching} in a graph is a set $M\subseteq E$ of its edges such that no two edges intersect, i.e., for each $e,e'\in M$ we have $|e\cap e'|\neq 1$. 
In a weighted graph, the weight of a matching $M$ is $\sum_{e\in M}w(e)$.
We also study the following extensions of matchings; note that they not always constitute a legal matching. 
A \emph{$b$-matching}~\cite{AhnG14} is a natural extension of the notion of matching, where each node can take part in at most $b$ edges of $M$.
A further extension of this gives a different bound $b_v$ to each node $v\in V$, as follows.
Given a graph $G=(V,E,w)$ and a vector $b$ where $b=(b_v)_{v\in V}$, a $b$-matching in $G$ is a set $M\subseteq E$ of edges such that for each $v\in V$, $|\{u\mid \{v,u\}\in M\}|\leq b_v$.
The weight of a $b$-matching is naturally defined as the sum of edge weights in $M$.
This definition also applies when $M$ is a multi-set.
A $b$-matching with the $b$-vector satisfying $b_v=1$ for all $v\in V$ is a (standard) matching.

A graph $G=(V,E)$ or $G=(V,E,w)$ is \emph{bipartite} if there is a bi-partition of its nodes, $V=L\cup R$, such that $E\subseteq L\times R$. 
A \emph{semi-matching} in a bipartite graph is a set $M\subseteq E$ of edges that intersects every node in $L$ exactly once. 
Semi-matchings in unweighted graphs were defined and discussed in~\cite{HarveyLLT06}, which presents a cost measure for them, aiming at capturing quantities related to task allocation problems.
In Section~\ref{subsec: semi matchings def} we discuss an extension of this definition to weighted graphs~\cite{BrunoCS74}.

\paragraph{Minimum weight spanning tree}
A spanning tree in a graph is a set $T\subseteq E$ of edges that intersects every node in $G$ at least once, and contains no cycles (a cycle is a non-trivial path from a node to itself that does not repeat any other node).
The weight of a tree $T$ in a weighted graph is the sum of its edge weights: $w(T)=\sum_{e\in T} w(e)$.
In the minimum weight spanning tree problem (MST), our goal is to find a spanning tree of minimum weight.

\paragraph*{Approximation algorithms}
For the shortest $(s,t)$-path problem, we consider an approximation algorithm with one sided error. 
That is, if a shortest $(s,t)$-path has length $\ell$, then the $(5/3-\delta)$-approximation algorithm we consider returns a path of length $\hat{\ell}$ satisfying $\ell \leq \hat{\ell} \leq (5/3-\delta)\ell$.
Specifically, this approximated value $\hat{\ell}$ allows us to distinguish the cases $\ell=3$ (in which case $\hat\ell<5$) and $\ell=5$ (in which case~$\hat\ell\geq5$).

\subsection{The \omv{} conjecture} 
We use the popular \omv{} conjecture as a condition for some of our lower bounds. 
In fact, we will not use it as is, but a related conjecture, called the \oumv{} conjecture.
We start by defining the \oumv{} problem.

\begin{definition}
	In the \emph{Online Boolean Vector-Matrix-Vector Multiplication (\oumv{})} problem,
	an algorithm is given an integer $n$ and an $n\times n$ Boolean matrix $M$. 
	Then, for $n$ rounds numbered $i=0,\ldots,n-1$, the algorithm is given a pair of $n$-dimensional Boolean column vectors $(u_i,v_i)$ and has to compute $u_i^\intercal Mv_i$ and output the resulting Boolean value before it can proceed to the next round. 
\end{definition}

It is conjectured~\cite{HenzingerKNS15} that the \oumv{} problem has no truly subcubic-time algorithm, as stated next.
This conjecture is implied by another, popular conjecture---the \omv{} conjecture.
Both conjectures implicitly assume there are $\Theta(n^2)$ non-zero entries in the matrix $M$.

\begin{conjecture}[The \oumv{} conjecture]
For any constant $\epsilon>0$, there is no $O(n^{3-\epsilon})$-time algorithm that solves \oumv{} with an error probability of at most~$1/3$.
\end{conjecture}

\section{Shortest $(s,t)$-paths}
\label{sec: sp}

The computation of shortest paths given the link weights,
is a most basic task in communication networks.
We will hence 
start our investigation of weight-dynamic algorithms with this use case.

We focus on maintaining the shortest $(s,t)$-path length
in an undirected graph.
Any algorithm that maintains distances between all pairs of nodes (all pairs shortest paths, or APSP) or between a source and all other nodes (single-source shortest paths, or SSSP) can also be used to compute the shortest $(s,t)$-path length, so our lower bounds apply to these problems as well.

In the edge-weight change regime, shortest $(s,t)$-path problem has a dynamic algorithm with $O(\sqrt{m}\log n)$ update time and constant query time for the path length, and $O(\ell)$ query time for the path itself, where $\ell$ is the length of the path~\cite{FrigioniMN00}. 
This problem does not have a truly sub-linear (in $n$) amortized time algorithm if edge insertions and deletions are both allowed~\cite{HenzingerKNS15}.
Here, we extend this lower bound construction to the case where only edge weight changes are allowed ---  note that the lower bound is on a graph with $m=\Theta(n^2)$.

\begin{theorem}
	\label{thm: shortest st path - lb}
	For any constants $\epsilon,\delta>0$,
	there is no dynamic algorithm maintaining a
	$(5/3-\delta)$-approximation for the shortest $(s,t)$-path problem with 
	$O(n^{1-\epsilon})$ amortized time per edge-weight change and
	$O(n^{2-\epsilon})$ amortized time per query,
	unless the \omv{} conjecture is false.
\end{theorem}

Given a Boolean $n\times n$ matrix $M$ and two Boolean $n$-dimensional vectors $u,v$, we define a weighted graph $G_{uMv}$ as follows.
The graph $G_{uMv}$ is a full bipartite graph on $V=(A\cup\{t\})\cup (B\cup\{s\})$, with $A=\{a_1,\ldots,a_n\}, B=\{b_1,\ldots,b_n\}$. 
The weights of the edges in $A\times B$ are $w(a_i,b_j)=3-2M[i,j]$ for every $(i,j)\in[n]^2$.
The weights of the other edges are $w(s,a_i)=3-2u[i]$ and $w(t,b_i)=3-2v[i]$ for every $i\in[n]$.
The following claim connects the \oumv{} problem with the shortest $(s,t)$-path problem.

\begin{claim}
	\label{claim: umv to shortest st paths}
	If $u^\intercal Mv=1$ then the shortest $(s,t)$-path in $G_{uMv}$ has weight $3$, and otherwise at least~$5$.
\end{claim}

\begin{proof}
	If $u^\intercal Mv=1$ then there exists $(i,j)\in[n]^2$ such that $u[i]=M[i,j]=v[j]=1$. In this case, the path $(s,a_i,b_j,t)$ has weight $3$.
	
	On the other hand, if $u^\intercal Mv=0$ then for each $(i,j)\in[n]^2$ either $u[i]=0, M[i,j]=0$, or $v[j]=0$.
	Consider a shortest $(s,t)$-path:
	It must start with an edge $(s,a_i)$ and end with an edge $(b_j,t)$ for some $(i,j)\in[n]^2$, and contain at least one other edge, in $A\times B$.
	The edges $(s,a_i)$ and $(b_j,t)$ have either weight~$1$ or~$3$. If either has weight $3$, then the path has weight at least $5$.
	Otherwise, we conclude that $u[i]=v[j]=1$, which implies $M[i,j]=0$, and thus $w(a_i,b_j)=3$.
	Since the graph is bipartite, any path from $a_i$ to $b_j$ must either use the edge $(a_i,b_j)$, or have length at least $3$; in both cases, this path has weight at least $3$, and the full $(s,t)$-path has weight at least~$5$.
\end{proof}

Using this claim, we can now prove Theorem~\ref{thm: shortest st path - lb}.

\begin{proof}[Proof of Theorem~\ref{thm: shortest st path - lb}]
	Consider an optimal dynamic algorithm for the shortest $(s,t)$-path problem.
	Using this algorithm we process pairs of vectors $u,v$ arriving as inputs to the \oumv{} problem with matrix $M$ as follows.
	
	First, build the graph $G_{uMv}$ for $u,v$ the all-$0$-vectors, and execute the initialization phase of the $s$-$t$-shortest paths algorithm.
	Given a pair $u,v$ of vectors arriving online for the \oumv{} problem, update the edge weights in the graph $G_{uMv}$. Note that the change is only in  edges touching $s$ or $t$, so $O(n)$ edges are changed, and only by an additive factor of $2$.
	
	Now, execute the algorithm on these updates, which takes time $T$. 
	If the algorithm returns a shortest path of length strictly less than $5$, we use Claim~\ref{claim: umv to shortest st paths} to conclude that $u^\intercal Mv=1$, and otherwise, $u^\intercal Mv=0$.
	Thus, we have solved the \oumv{} problem by performing $O(n)$ changes and a single query for the simulation of each round, and a total of $O(n^2)$ changes and $n$ queries. 
	If the amortized update time of each change is $O(n^{1-\epsilon})$ and of each query is $(n^{2-\epsilon})$, then the total running time would be $O(n^{3-\epsilon})$, contradicting the \oumv{} conjecture.
\end{proof}

\section{Maximum $(s,t)$-flow}
\label{sec: max flow}

Maximum flow (or throughput) problems are among the most studied
network optimization problems, and we hence consider
weight-dynamic algorithms for them.

\subsection{Maximum $(s,t)$-flow algorithms}
\label{subsec: max flow alg}

\begin{theorem}
	\label{thm: mx st flow with only weight increases - alg}
	There is a dynamic algorithm for maintaining the maximum $(s,t)$-flow 
	in directed or undirected graphs with edge weight increases and decreases, 
	with constant additive changes, 
	in $O(m)$ time per operation.
\end{theorem}

%

When considering both edge insertions and deletions, no non-trivial algorithm for maintaining max $(s,t)$-flow is known.
The best strategy for such cases is to execute the state of the art algorithm for static graphs, requiring $O(m^{10/7}W^{1/7})$ rounds to complete~\cite{Madry16}.
Our algorithm constitutes a significant improvement over this strategy.

\begin{proof}
	Let us first consider an edge-weight (capacity) change by a single unit. 
	
	We start with the case of a weight \emph{decrease}, for some edge $(u,v)$.
	If this edge was not saturated before the weight decrease, no change is required.
	Otherwise, if the graph is undirected assume w.l.o.g. that the flow on this edge goes from $u$ to $v$, and if it is directed, similarly assume the edge is directed from $u$ to $v$.
	Consider the directed graph composed of all the flows.
	In this graph, we look for a directed path from $s$ to $u$, and for a directed path from $v$ to $t$; both paths must exist since the flow is valid.
	We then decrease the flow on these two paths by $1$ unit, as well as the flow on the edge $(u,v)$ itself.
	Finally, search for an $(s,t)$-path in the residual graph, and if it is found, increase the flow on it by one unit.
	
	Similarly, consider a weight \emph{increase} on an edge $(u,v)$.
	If the graph is directed, assume w.l.o.g. that the edge is directed from $u$ to $v$. 
	In the residual graph, search for two paths: from $s$ to $u$ and from $v$ to $t$. If these are found, increase the flow on them and on $(u,v)$ by $1$ unit.
	If the graph is undirected, apply the procedure twice, with the directed edges~$(u,v)$ and~$(v,u)$.
	
	Each of the above procedures takes $O(m)$ time.
	If the weight change is by more than one unit, we just repeat the same procedure for each unit change, which still takes $O(m)$ time, as claimed.
\end{proof}

\subsection{Lower bounds for maximum $(s,t)$-flow}
\label{subsec: max flow lb}

Dahlgaard~\cite{Dahlgaard16} has proved an $\Omega(n^{1-\epsilon})$ lower bound on the amortized update time of the maximum cardinality bipartite matching problem, conditioned on the \omv{} conjecture.
From this, he derives a lower bound for the maximum $(s,t)$-flow problem using standard techniques.

Dahlgaard's lower bound is proved using a construction of a sequence of specifically crafted bipartite graphs. 
Here, we adapt this lower bound to our setting, by using the following result. This is~\cite[Lemma~1]{Dahlgaard16}, along with an observation regarding the number of changes, which is implicit in the graph construction therein.

\begin{lemma}[\cite{Dahlgaard16}]
	\label{lem: oumv to mcm from Dahlgaard16}
	Given an instance $M, (u_i,v_i)_{0\leq i\leq n-1}$ of the \oumv{} conjecture, it is possible to construct a sequence $(G_i)_{0\leq i\leq n-1}$ of bipartite graphs such that if $u_i^\intercal  Mv_i=1$ then the maximum size of a matching in $G_i$ is $4n+2i+1$, and otherwise it is $4n+2i$. 
	In addition, all the graphs have the same bi-partition, with $6n$ nodes in each side, and the difference between $G_i$ and $G_{i+1}$ is in the addition of $O(n)$ edges.
\end{lemma}

From this, we can derive the following lemma.

\begin{lemma}
\label{lem: oumv to mwm}
	Given an instance $M, (u_i,v_i)_{0\leq i\leq n-1}$ of the \oumv{} conjecture, it is possible to construct a sequence $(H_i)_{0\leq i\leq n-1}$ of weighted full bipartite graphs such that if $u_i^\intercal  Mv_i=1$ then the maximum weight of a matching in $H_i$ is $10n+2i+1$, and otherwise it is $10n+2i$.
	In addition, all the graphs have $12n$ nodes, and the difference between $H_i$ and $H_{i+1}$ is in the increase of $O(n)$ edge weights from $1$ to $2$.
\end{lemma}

\begin{proof}[Proof of Lemma~\ref{lem: oumv to mwm}]
	As the graphs from Lemma~\ref{lem: oumv to mcm from Dahlgaard16} do not differ in their set of nodes or bi-partition, we can consider a full bipartite graph $H$ on the same set of nodes and bi-partition, with all edge weights $1$.
	Let $H_i$ be the graph $H$, with the weight of edges appearing in $G_i$ increased by $1$.
	Note that the matching $M_i$ from $G_i$ also appears in $H_i$, and its weight $w_{H_i}(M_i)$ is doubled, i.e., $w_{H_i}(M_i)=2|M_i|$.
	In addition, this matching matches $2|M_i|$ nodes, and we can also match all the other nodes, since the graph is full bipartite, getting a matching $M'_i$ of total weight
	$W_{H_i}(M'_i)= 2|M_i| + 6n - |M_i|= 6n + |M_i|$.
	
	Finally, note that no matching in $H_i$ can contain more weight-$2$ edges (since $M_i$ was of maximum size in $G_i$), or more edges in total (since $M'_i$ matches all the nodes). Hence, $M'_i$ is a maximum weight matching in~$H_i$.
\end{proof}

Assume there is a dynamic algorithm that solves the MWM problem on a sequence of graphs with edge weight increases, in $O(n^{1-\epsilon})$ amortized time per weight increase and  $O(n^{2-\epsilon})$ amortized time per query. 
The above lemma shows that, given an instance for the \oumv{} conjecture of length $n$, we can construct the sequence $(H_i)_{0\leq i\leq n-1}$ of graphs, and use the alleged algorithm to solve the instance of the \oumv{} conjecture.
In addition, the graphs have $O(n)$ nodes and $O(n^2)$ edges.
So this gives an algorithm solving the \oumv{} problem in $O(n^{2-\epsilon})$ time per vector-pair, refuting the conjecture. 
As the \omv{} conjecture implies the \oumv{} conjecture, we get the following theorem.

\begin{theorem}
	\label{thm: mwm with only weight increases}
	There is no dynamic algorithm maintaining a maximum weight matching in
	$O(n^{1-\epsilon})$ amortized time per weight increase and
	$O(n^{2-\epsilon})$ amortized time per query,
	for any $\epsilon>0$,
	unless the \omv{} conjecture is false.
\end{theorem}

As $b$-matchings constitute a specific case of matchings, the theorem immediately applies to them as well.

A standard textbook reduction, as the one mentioned in~\cite{Dahlgaard16} (see, e.g.,~\cite{CormanLRS_09}), gives a similar bound for the max $(s,t)$-flow, as follows.

\begin{theorem}
	\label{thm: mx st flow with only weight increases - lb}
	There is no dynamic algorithm for maintaining a maximum $(s,t)$-flow 
	in a directed graph with edge weight increases in
	$O(n^{1-\epsilon})$ amortized time per weight increase and
	$O(n^{2-\epsilon})$ amortized time per query,
	for any $\epsilon>0$
	unless the \omv{} conjecture is false.
\end{theorem}

The same results hold for the decremental case (and of course, if both weight increases and decreases are allowed). 
To see this, consider the same graph sequences, but in the opposite order, in the same way as in~\cite{Dahlgaard16}.

Following the previous section, one may wonder if it is possible to get similar lower bounds for approximate versions of the maximum flow and maximum matching problems.
At least in one case, we know this is impossible: there exist $2$-approximation dynamic algorithms for matching, running in constant time~\cite{Solomon2016FullyDM}.
Since these dynamic algorithms can also be used in our setting, getting a non-trivial lower bound for $2$-approximation of maximum matching is impossible.
The questions of lower approximation factors of matching,
and of getting a lower bound for maximum $(s,t)$-flow not through a reduction to matching, remain open.

\section{Matchings}
\label{sec: matchings}

We next consider applications related to load-balancing,
and specifically, matchings.
In the previous section, we showed a lower bound for maximum matching, which immediately implies a lower bound for the $b$-matching problem. 
Here, we present a maximum matching algorithm, and then turn to study a related problem: the computation of semi-matchings. Semi-matchings are traditionally studied in the context of allocating tasks to machines~\cite{HarveyLLT06}, but also find applications, e.g., when assigning users
to points-of-presence~\cite{SchmidS13}.

\subsection{Maximum weight matching algorithm}
\label{subsec: mwm - alg}
The aforementioned dynamic algorithm for the maximum $(s,t)$-flow problem immediately implies a similar algorithm for the dynamic maximum weight matching problem, through the same standard reduction used for the lower bound (see e.g.~\cite{CormanLRS_09}).
This gives the following.

\begin{theorem}
	\label{thm: mwm - alg}
	There is a dynamic algorithm for maintaining a maximum weight matching in graphs with edge weight increases and decreases, 
	with constant additive changes, 
	in $O(m)$ time per operation.
\end{theorem}

\subsection{Defining weighted semi-matchings}
\label{subsec: semi matchings def}

Consider a weighted bipartite graph $G=(L,R,E,w)$ where $E\subseteq L\times R$ is the set of edges and $w:E\to\mathbb N$ is an edge-weight function,
where no node in $L$ is isolated.
A \emph{semi-matching} is a set $M\subseteq E$ of edges in $G$ such that each node in $L$ is incident on exactly one edge of~$M$.

Following~\cite{BrunoCS74}, we study an extension of this classical definition of semi-matchings to the edge-weighted case, 
which can model task allocation (resp. PoP assignment) problems where a task may have different completion times in different machines (resp. at different sites). 
We consider a model where the \emph{cost} of a semi-matching $M$ is defined to capture the total makespan (sum of finishing times) of all the tasks (nodes in~$L$), or equivalently, the mean makespan (average finishing time).
Since the tasks are no longer identical, each machine sets the order of the tasks allocated to it, in a way that will minimize the total makespan.
Formally, we define the cost of $M$ for a node (machine) $r\in R$ as
\[
\cost_G(M,r)=
\min_{\pi:[\deg_M(r)]\to \Gamma_M(r)}
\sum_{i=1}^{\deg_M(r)} \sum_{j\leq i} w(\pi(j),r)
\]
where $\Gamma_M(r)$ is the set of neighbors of $r$ in~$M$,  $\deg_M(r)=|\Gamma_M(r)|$, and $[\deg_M(r)]=\{1,\ldots,\deg_M(r)\}$.
The total cost of a matching~$M$ is the sum of the costs of nodes in $R$, i.e., $\cost_G(M)=\sum_{r\in R}\cost_G(M,r)$.
Note that this definition coincides with the definition for unweighted graphs~\cite{HarveyLLT06}, 
by setting $w$ as the constant function $1$.

In the \emph{minimum-cost semi-matching} problem, we are given a weighted bipartite graph as above, and the goal is to find in it a semi-matching with a minimal cost.

While semi-matchings resemble $b$-matchings in structure, finding minimum cost matchings of these types constitutes a different algorithmic problem. 
One reason for this is that in a semi-matching, the effect of a matched edge on the cost is not linear, but also depends on the other edges. Another reason is that in semi-matchings, there is no predefined bound on the number of matching edges touching each node.

\subsection{Lower bounds for the semi-matchings problem}
\label{subsec: semi matchings lb}

The lower bound for the maximum weight matching problem from Theorem~\ref{thm: mwm with only weight increases} also implies a lower bound for maintaining a minimum cost semi-matching as follows.

\begin{theorem}
	\label{thm: semi-matchings - lb}
	There is no dynamic algorithm for maintaining a minimum cost semi-matching in
	$O(n^{1-\epsilon})$ amortized time per weight increase and
	$O(n^{2-\epsilon})$ amortized time per query,
	for any $\epsilon>0$,
	unless the \omv{} conjecture is false.
\end{theorem}

To prove this theorem, we use Lemma~\ref{lem: oumv to mcm from Dahlgaard16}, in a way similar to the one taken in the proof of Lemma~\ref{lem: oumv to mwm}.
Specifically, given a graph $G_i$ from Lemma~\ref{lem: oumv to mcm from Dahlgaard16}, we construct a graph $H_i$ as follows.

Let $H$ be a full bipartite graph on the same set of nodes as $G_i$ and the same bi-partition (all these graphs have the same nodes and partition), with all edge weights set to $2$.
Let $H_i$ be the graph $H$ with the weights of edges that appear in $G_i$ decreased by $1$.
A matching $M_i$ in $G_i$ is also a matching in $H_i$, with cost  $\cost_{H_i}(M_i)=|M_i|$.
In addition, this matching matches $|M_i|$ nodes from $L$, and we can also match all the rest of the $L$-nodes without using any of the nodes of $R$ twice, as the graph is full bipartite with equal-sized sides.
This matching $M_i'$ is a semi-matching, with a cost
$\cost_{H_i}(M_i')=|(M_i))| + 2(6n - |M_i|)= 12n - |M_i|$.
This gives an upper bound on the minimum weight of a semi-matching in $H_i$.
Finally, note that a semi-matching in $H_i$ cannot use less edges, and cannot use less of weight $2$ in $H_i$, so its cost cannot be smaller.
Hence, the minimum cost of a semi-matching in $H_i$ is $12n - |M_i|$, yielding the next lemma.

\begin{lemma}
	\label{lem: oumv to semi matchings}
	Given an instance $M, (u_i,v_i)_{0\leq i\leq n-1}$ of the \oumv{} conjecture, it is possible to construct a sequence $(H_i)_{0\leq i\leq n-1}$ of weighted full bipartite graphs such that if $u_i^\intercal  Mv_i=1$ then the cost of a semi-matching in $H_i$ is at least $8n-2i+1$, and otherwise it is $8n-2i$.
	In addition, all the graphs have $12n$ nodes, and the difference between $H_i$ and $H_{i+1}$ is in the decrease of $O(n)$ edge weights from $2$ to $1$.
\end{lemma}

This lemma immediately implies Theorem~\ref{thm: semi-matchings - lb}, in a way similar to the derivation of Theorem~\ref{thm: mwm with only weight increases} from Lemma~\ref{lem: oumv to mwm}.

\section{Minimum weight spanning trees}
\label{sec: mst}

We conclude with another important application in communication networks:
the computation of spanning trees. These are used, e.g.,
on the Layer~2 of the networking stack to avoid forwarding loops, and also find applications in the context of sensor networks~\cite{huang2006dynamic}.

\subsection{Minimum weight spanning tree construction algorithms}
\label{subsec: mst alg}
Holm et al.~\cite{HolmRW15} gave an algorithm that maintains a minimum spanning tree in a fully dynamic setting, i.e., with both edge insertions and deletions, in $O(\log^4 n /\log \log n)$ amortized time per update, and constant time per MST-size query.
This algorithm can be adapted to our setting without any increase in the asymptotic running time: any edge-weight change in our model is simply translate to the deletion of the edge with its old weight, and addition of it with its new weight.

\subsection{Lower bounds for minimum weight spanning tree computation}
\label{subsec: mst lb}

In their seminal works, Pătraşcu and Demaine~\cite{PatrascuD06} gave a cell-probe lower bound for the connectivity problem in a dynamic graph.
More concretely, they consider a data structure supporting three operations: adding an edge between two nodes, removing an edge, and deciding whether the graph is connected.
In this setting, they prove that maintaining any such data structure requires  $\Omega(\log n)$ amortized complexity per operation.
The complexity measure used in their work is the number memory words accessed, which implies a similar lower bound on the amortized running time.

\begin{theorem}[{\cite[Thm.~2.4 \& \S 9.1]{PatrascuD06}}]
	\label{thm: connectivity lb PD06}
	There is no dynamic algorithm, even randomized, 
	maintaining graph connectivity  in an unweighted graph 
	performing $o(\log n)$ cell probes amortized by edge change
	and $o(\log n)$ cell probes per query, 
	with cells of size~$O(\log n)$. 
\end{theorem}

We prove our lower bound through a reduction to the bound of Pătraşcu and Demaine, without getting into the details of their construction.

\begin{theorem}
	\label{thm: mst - lb}
	There is no dynamic algorithm, even randomized,
	 maintaining minimum weight spanning tree  
	in $o(\log n)$ amortized time per edge-weight change and
	$o(\log n)$ amortized time per query.
\end{theorem}

\begin{proof}
	Consider a data structure $MST$ that supports two operations: setting the weight of a given edge, and checking the MST size.
	We construct a dynamic graph $H$ maintained by this data structure, and show that if this data structure can be maintained in $o(\log n)$ cell-probe operations,
	and answer queries in $o(\log n)$ cell-probe operation,
	then so does the connectivity of an unweighted graph $G$, contradicting Theorem~\ref{thm: connectivity lb PD06}.
	
	We design a data structure for connectivity, that supports edge addition and deletion and connectivity queries, as follows.
	Let $G$ be an unweighted, dynamic graph, which is initially empty. 
	Our data structure will maintain a graph $H$, which is a complete weighted graph on the same set of nodes as $G$.
	All the edges of $H$ start with weight $2$.
	Whenever an edge $(u,v)$ is added to $G$, the weight of the corresponding edge $(u,v)$ in $G$ is reduced to $1$. When an edge is removed from $G$, the weight of the corresponding edge in $H$ is increased back to $2$.
	Finally, for a connectivity query in $G$, we make an MST-size query in $H$, and answer on the affirmative iff the MST size is $n-1$.
	
	The query complexity lower bound is now immediate, as every operation on $G$ is translated to a single operation on $H$. To show correctness note the following invariant: the edges of $G$ are exactly the edges that have weight $1$ in $H$, while all other edges of $H$ have weight $2$.
	Now, note that if $G$ is connected, then it has a spanning tree of $n-1$ edges; in $H$, the same edges form a spanning tree of weight $n-1$.
	On the other hand, if $H$ has a spanning tree of weight $n-1$, then this spanning tree must consist solely of edges of weight $1$, i.e., edges that exist in $G$.
	The same spanning tree thus appears in $G$, which is therefore connected, completing the proof.
\end{proof}

\section{Related Work}
\label{sec:relwork}

We have recently witnessed many great efforts to
render networks more flexible and adaptive, from the application
layer, over the transport layer, network layer,
link layer, down to the physical layer, to just reference some examples~\cite{pieee19}.
For the specific case study of traffic engineering considered in this paper, 
empirical works show that adjusting OSPF weights and hence rerouting traffic multiple times per day, allows to improve network efficiency and account 
for time-of-day effects and to some extent circadian effects~\cite{b4,fischer2006replex}.
Accordingly, there has recently been much interest in more dynamic traffic engineering mechanisms~\cite{fischer2006replex,kumar2018semi}, which also find first deployments~\cite{b4}. 
That said, the earliest approaches date back to the Arpanet~\cite{khanna1989revised}.
While the high runtime of traffic engineering algorithms was often a main concern in the literature~\cite{traffic-engineering1,vissicchio2014sweet,b4}, we are not aware of any work on \emph{dynamic} algorithms to improve performance, which is the focus of this paper.

Dynamic graph algorithms have received much attention
in the literature. 
In the following, we briefly review related work 
on fully dynamic graph algorithms (both edge
insertions and deletions), and refer the reader to surveys on the topic for a further discussion~\cite{Henzinger18,DBLP:journals/jda/DemetrescuI06}. 

For the shortest path problem there is a lower bound conditioned on the \omv{} conjecture showing that for any $\epsilon > 0$  and $\delta > 0$ no $(5/3-\delta)$-approximation algorithm for the shortest $(s,t)$ problem can achieve
$\Omega(n^{1/\epsilon})$ amortized time per edge insertion, edge deletion, or query~\cite{HenzingerKNS15}. This lower bound even holds in \emph{unweighted, undirected} graphs. There are also already many interesting upper bounds in different scenarios, depending on whether the graph is directed or undirected, weighted or unweighted, whether the running time is amortized or worst-case, and whether the adversary that generates the sequence of operations is oblivious or adaptive. We restrict ourselves here to amortized running against an adaptive adversary and the results we state apply to weighted, directed and undirected graphs. The fastest fully dynamic exact all-pairs-shortest path algorithm takes time $O(n^2 (\log n + \log^2 (m/n)))$ per edge update operation and constant time per query~\cite{DBLP:journals/jacm/DemetrescuI04,DBLP:conf/swat/Thorup04}.
For single-source shortest paths the fastest exact fully dynamic algorithm  is still $O(m)$; 
in  graphs with real edge weights in $[1, W]$ there exists a $(1+\epsilon)$-approximation algorithm in time $O(n^{1.823}/ \epsilon^2)$ per update for any small $\epsilon > 0$ using fast matrix multiplication and $O(n^{2.621})$ preprocessing time~\cite{van2019dynamic}.
There are also algorithms that realize interesting tradeoffs between approximation ratio and query time~\cite{DBLP:journals/corr/abs-2004-10319}, e.g., guaranteeing a $O(\log^4 n)$-approximation with $O(m^{1/2 + o(1)})$ time per operation, or a $n^{o(1)}$-approximation with $O(n^{o(1)})$ update time.

Dynamicity restricted to edge-weight changes was much less studied. Single source shortest paths can be maintained in this case in $O(\log n)$ worst-case update time in specific graph families, such as graphs with bounded genus, and in $O(\sqrt{m}\log n)$ update time for general graphs, all with constant query time~\cite{FrigioniMN00}.
APSP can be maintained in graphs of treewidth $k$ in $O(k^3\log n)$ update time and $O(k^2\log(n)\log(k\log n))$ query time~\cite{AbrahamCDGW16}.
A $(1+\epsilon)$-approximation for distance queries (i.e., APSP), can be maintained in planar graphs when the edge-weight changes are such that no distance in the graph is stretch by more than a multiplicative factor $M$, in $\tilde{O}(M^4/\epsilon^3)$ update and query time~\cite{AbrahamCDGW16}.

There exists no exact fully dynamic $s$-$t$ flow algorithm that is faster than recomputation from scratch, and there is a conditional lower bound based on the \omv{} conjecture of $\Omega(n^{1 - \epsilon})$ time per operation~\cite{Dahlgaard16} and also based on other conjectures~\cite{abboud2018matching}. 
For the case of only edge-weight changes, there is an algorithm that has the same worst-case guarantees as recomputation, but better performance on real-world cases~\cite{GoldbergHKKTW15}.
The fully dynamic algorithm we present here also leads to an insertions-only or deletions-only algorithm with amortized time $O(n)$ per operation~\cite{DBLP:journals/corr/abs-1804-01823}. Nonetheless, various approximation algorithms exist: An $O(\log n \log \log n)$-approximation in time $\tilde O(m^{3/4})$~\cite{DBLP:journals/corr/abs-2005-02368}
and a $n^{o(1)}$-approximation in $O(n^{o(1)})$ update time~\cite{goranci2020expander}.

Finally, there exists an $\Omega(\log n)$ lower bound on the time per operation for any fully dynamic minimum spanning tree algorithm~\cite{PatrascuD06}.
Regarding upper bounds, there is an $O(\log^4 n /\log \log n)$ time exact fully dynamic minimum spanning tree algorithm~\cite{HolmRW15}
and a $O(\log n (\log \log n)^2)$-time algorithm for connectivity~\cite{HuangHKP17} that can be turned into a $(1+\epsilon)$-approximate minimum spanning tree algorithm by partitioning the edges into weight classes and maintaining for each weight class a spanning forest that contains all edges of the spanning forests of the smaller weight class.

However, we are not aware of any study on dynamic
graph algorithms for link weight changes.

\section{Conclusion}
\label{sec:conclusion}

Motivated by networking
applications in which resource allocations depend on 
(and are controlled
by) link weights, we initiated the study of
weight-dynamic network algorithms. 
Considering different applications, from
traffic engineering to spanning tree constructions,
we derived bounds on the potential speedup of
dynamic algorithms compared to computation from scratch.

While most traffic engineering solutions today 
are still designed for fairly static scenarios, 
this can come at the price
of a suboptimal resource utilization or performance: by not reacting to
specific changes in the demand, a network may temporarily become overutilized
(harming performance) or underutilized (harming efficiency), depending on the situation.
Since traffic patterns typically indeed feature much temporal
structure, this is problematic.
We hence believe that our study aligns well with the 
current trend toward more adaptive
communication networks,
and can inform the networking community what can and cannot be achieved with such dynamic algorithms.
We further believe that our notion of 
``weight-dynamic'' network algorithms may also be of interest to the theory community, and can lead to interesting follow-up work.

\bibliographystyle{plain}
\bibliography{biblio}

\end{document}